\documentclass{llncs}

\usepackage{amssymb}
\usepackage{amsmath}
\usepackage{graphicx}

\usepackage{color}
\usepackage[usenames,dvipsnames,svgnames,table]{xcolor}

\usepackage{bbm}

\usepackage[OT4]{fontenc}
\usepackage[utf8x]{inputenc}

\usepackage[normalem]{ulem} 
\usepackage{soul} 

\usepackage{listings}

\newcommand{\mbraket}[2]{\langle #1 | #2 \rangle}

\newcommand{\mtr}[1]{\mathrm{Tr}\left( #1 \right)}

\newcommand{\myiff}{{\bf iff }}

\newcommand{\bC}{\mathbb{C}}

\newcommand{\nH}{\mathcal{H}}
\newcommand{\lH}{\mathfrak{h}}

\newcommand{\nU}{\mathcal{U}}

\newcommand{\keywords}[1]{\par\addvspace\baselineskip\noindent\keywordname\enspace\ignorespaces#1}

\begin{document}


\mainmatter

\title{Maximally entangled states of $(d,\infty)$ quantum systems: some numerical studies}



\authorrunning{Roman Gielerak \and Marek Sawerwain}
\tocauthor{Roman Gielerak and Marek Sawerwain}

\author{Roman Gielerak\inst{1} \and Marek Sawerwain \inst{1}}
\institute{Institute of Control \& Computation Engineering \\
University of Zielona G\'ora, Licealna 9, Zielona G\'ora 65-417, Poland \\
\email{R.Gielerak@issi.uz.zgora.pl}
\and
Institute of Control \& Computation Engineering \\
University of Zielona G\'ora, Licealna 9, Zielona G\'ora 65-417, Poland \\
\email{M.Sawerwain@issi.uz.zgora.pl}
}


\maketitle

\begin{abstract}
Gram matrix approach to an entanglement analysis of pure states describing  $(d,\infty)$ -- quantum systems is being introduced. In~particular, maximally entangled states are described as those having a special forms of the corresponding Gram matrices.
\keywords{quantum systems, entanglement, numerical computations}
\end{abstract}

\section{Introduction} \label{lbl:sec:introduction:RG:MS:CN:2019}

One of the most important, genuine quantum resource in quantum information engineering is by no doubt, the phenomenon of quantum entanglement of the quantum states being processed \cite{Bengtsson2006}, \cite{Guhne2008}, \cite{Horodecki2009}. The known and implemented quantum communication protocols, including the best known QKD \cite{Horodecki2010} and teleportation protocols, are heavily based on the entanglement present in the corresponding quantum states used in a particular implementation of the protocol performed. This is the main reason why the mathematical analysis, together with the corresponding engineering of quantum entanglement is so important for future developments of quantum realistic technologies \cite{Castelvecchi2017}. 

There are many excellent sources on the mathematical foundations of the quantum information processing \cite{Nielsen2000} to which we refer for details.

Concerning the existing quantum technologies the major problem is to find quantum systems in which genuine quantum behaviour can survive on a sufficient long time interval due to intense, unavoidable decoherence processes destroying them. There are many proposals, together with preliminary technological implementations for the construction of the quantum computer, see i.e. \cite{DiVincenzo2000}. At present  days  many  of world leading  technological  companies, like IBM \cite{IBMQ}, Google \cite{Google}, Rigetti \cite{Rigetti},  D-Wave \cite{DWave}  and  others started  seriously to develop  full scale quantum computers hardware and software \cite{Intel}, \cite{Qiskit}, \cite{Microsoft}. The  use  of  superconducting qubits \cite{Devoret2004} together with Josephson junctions \cite{PhysRevLett.111.080502} properties in performing universal quantum  computations seems to be extremely difficult challenge from the engineering point of view. And this is why any alternatives to the proposed in \cite{IBMQ} quantum computer hardware technologies are so important \cite{Linke2017}.

The class of quantum systems which can be seen as consisting of two entangled parts, one being finite dimensional (spin-degrees of freedom) and the second infinite dimensional, the so called quantum $(d,\infty)$ -- systems seems to be a very optimistic candidates for the realistic use in quantum processing information future technological implementations \cite{Gielerak2017}. Also in fundamental atomic and molecular quantum physics the quantum entanglement,  in between  discrete degrees of  freedom that given by continuous variables and called there spin-orbit entanglement is under intensive studies presently \cite{WitczakKrempa}, \cite{Karimi}.

The paper is organized as follows: in Sec.~\ref{lbl:sec:gram:matrix}  the  notion  of Gram matrix connected  to  a  given  pure  state  of  the  $(d, \infty)$  system  under   consideration  will  be  explained  and  some basic  properties of  it are  listed. Also our main  hypothesis (stating that the maximal entanglement  of   the  pure state  under  consideration  is manifested  in the special  form  of  the  corresponding  Gram matrix) is  formulated  there. Section~\ref{lbl:sec:title:RG:MS:CN:2019}  contains  some numerical  results in favour of the validity of our main  hypothesis.

\section{The Gram matrix approach to the entanglement of pure states of $(d,\infty)$ class of systems}
\label{lbl:sec:gram:matrix}

Let $d$ be a positive integer, $d \geq 2$ and let $\mathbb{C}^d$ be complex d-dimensional Euclid space. As it is well known \cite{Nielsen2000} pure states of any quantum system are given by unit vectors (by rays -- more precisely) in some separable Hilbert space, to be denoted as $\nH$. A composite two-partite quantum systems composed from the finite dimensional part, described by $\mathbb{C}^d$ and the infinite-dimensional (in general everything below applies to the finite dimensional systems as well) ones is described by the Hilbert space $\nH = \mathbb{C}^d \otimes \lH$, where $\lH$ stands for the appropriate, infinite  dimensional in  general, separable  Hilbert  space  and operator $\otimes$ stands for tensor product.

Let $(e_i)_{i=1:d}$ be the system of unit vectors in $\mathbb{C}^d$ defined as $(e_i)_j = \delta_{ij}$ for $i,j=1:d$. We denote as $(c)ons(\lH)$ the set of all (complete) orthonormal systems $(f_i)_{i=1,...}$ of vectors from $\lH$. For a given $(f_i) \in cons(\lH)$ and any $\Psi \in \nH$ we can write
\begin{equation}
\Psi = \sum_{i=1:d, j=1,...} c_{ij} e_i \otimes f_j = \sum^{d}_{i=1} e_i \otimes \left( \sum_{j=1}^{\infty}c_{ij}f_{j} \right) = \left[ \begin{array}{c}
	{\psi_1} \\
	\vdots \\
	{\psi_{d}} 
	\end{array} \right],
	\label{lbl:psi:decomposition:RG:MS:2018}
\end{equation}
where ${\psi_i} = \sum_{j=1}^{\infty} c_{ij} f_j$.

The vectors ${\psi_i}$ in general are not forming orthonormal systems of vectors.  From the orthonormality of systems $(e_i) \in cons(\mathbb{C}^d)$ the following normalisation condition follows 
\begin{equation}
{|| \Psi ||}^2_{\mathbb{C}^d \otimes \lH} = \sum_{i=1}^{d} {|| {\psi_i} ||}^2_{\lH} .
\end{equation}

For any $\Psi \in \mathbb{C}^d \otimes \lH$, the following,  finite Schmidt decomposition 
can be derived by means of the well known SVD arguments:
\begin{equation}
\Psi = \sum_{i=1}^{d} \sqrt{\lambda_i} \hat{e_i} \otimes \hat{f_i}
\label{lbl:eq:schmidt:RG:MS:2019}
\end{equation}
where $\lambda_i \geq 0$ for $i=1:d$, $(\hat{e_i}) \in cons(\mathbb{C}^d), (\hat{f_i}) \in ons(\lH)$.

All the information about quantum entanglement contained in $\Psi$ can be read off from the Schmidt decomposition~Eq.~(\ref{lbl:eq:schmidt:RG:MS:2019}). In particular:
\begin{itemize}
\item[(1)] The amount of entanglement (that will be computed by the amount of von Neumann entropy of the corresponding reduced density matrix which is the same as the  amount of entropy of the introduced below Gram matrix):
\begin{equation}
\mathrm{Ent}( \Psi ) = - \sum_{i=1}^{d} \lambda_i \log( \lambda_i ),
\label{lbl:eq:Ent:for:Psi}
\end{equation}
where $0 \cdot \log( 0 ) = 0$ convention has been used in formula~(\ref{lbl:eq:Ent:for:Psi}).
\item[(2)] The maximally entangled, pure states of the $(d,\infty)$ quantum systems are those for which for $i=1:d$, $\lambda_i=\frac{1}{d}$ and then $Ent(\Psi) = \log(d)$.
\end{itemize}

The problem of computing the Schmidt decomposition~Eq.~(\ref{lbl:eq:schmidt:RG:MS:2019}), see \cite{Gielerak2017} for example of a given $\Psi \in \mathbb{C}^{d} \otimes \lH$, is not easy in general \cite{You2015}. Here we propose another approach, to the computation of the Schmidt's numbers of a given $\Psi \in \mathbb{C}^{d} \otimes \lH$.

Let $\Psi \in \mathbb{C}^{d} \otimes \lH$ be given as in~Eq.~(\ref{lbl:psi:decomposition:RG:MS:2018}) i.e.:
\begin{equation}
\Psi = \sum_{i=1}^{d} e_i \otimes \psi_i .
\label{lbl:eq:Psi:form}
\end{equation}

Define the map $j_{\Psi} : \mathbb{C}^d \rightarrow \lH$ by $j_\Psi(e_i) = \psi_i$ and then extend this by linearity to the whole space $\mathbb{C}^d$, i.e.:
\begin{equation}
j_{\Psi} \left( \sum_{i=1}^{d} d_i e_i \right) = \sum_{j=1}^{d} d_i \psi_i .
\end{equation}

The map $j_{\Psi}$ is continuous and $|| j_{\Psi} || = || \Psi ||$, where $|| \cdot ||$  stands for the corresponding operator, respectively vector norm. From the very definition of $j_{\Psi}$ one can compute the operator:
\begin{equation}
G(\Psi) : j^{\dagger}_{\Psi} \cdot j_{\Psi} : \mathbb{C}^d \longrightarrow \mathbb{C}^d ,
\end{equation}
as a $d \times d$ matrix (which is the matrix representation of $G$ in the canonical basis of $\mathbb{C}^d$) with the corresponding matrix elements:
\begin{equation}
G( \Psi )_{ij} = \mbraket{\psi_i}{\psi_j} \;\;\; \mathrm{for} \;\;\; i,j=1:d.
\label{lbl:G:of:Psi}
\end{equation}

\begin{definition}
For a given finite set $\Omega = \{\omega_1, \omega_2, \ldots, \omega_n\} \subseteq \lH$ the Gram matrix of $\Omega$ is defined as $G(\Omega)_{ij} = \mbraket{\omega_i}{\omega_j}_{i,j = 1:n}$.
\end{definition}

Some elementary properties of the Gram matrix are listed below:
\begin{itemize}
\item[G(1)] The Gram matrix $G(\Omega)$ is hermitean and non-negative, i.e.:
\begin{equation}
\mathrm{for} \; i,j=1 : n : \overline{G_{ij}} = G_{ij}
\end{equation}
and, for any sequence of complex numbers $c_{i,j} \;  i=1:n$.
\begin{equation}
\sum_{i,j=1}^{n} c_{i} \overline{c_j}  G_{ij} \geq 0 .
\end{equation}
\item[G(2)] The rank of $G(\Omega)$ is equal to the dimension of the subspace $lh(\omega_1, \omega_2, \ldots, \omega_n) \subseteq \lH$, i.e.:
\begin{equation}
\mathrm{rank} \; G(\Omega) = \dim \; lh(\omega_1, \omega_2, \ldots, \omega_n) .
\end{equation}
\item[G(3)] If $G(\Omega)$ is a Gram matrix, then for all $i,j$:
\begin{equation}
| G_{ij}(\Omega) | \leq G_{ii}(\Omega) \cdot G_{jj}(\Omega) .
\end{equation}
\begin{proof}
By direct use of a Cauchy-Schwartz inequality.
\end{proof}
\item[G(4)] the Gram matrix $G(\Omega)$ is invertible iff  $\mathrm{rank} \; G(\Omega)=n$.
\item[G(5)] If a given $d \times d$ matrix $G$ is positive semi-definite then there exists system of vectors $(g_1, \ldots, g_d) \in \mathbb{C}^d$ such that $G$ is the Gram matrix of the system $(g_1, \ldots, g_d)$.
\begin{proof}
To compute the Cholesky decomposition of $G$, i.e. compute the matrix $B$ (lower-triangular or equivalently  upper-triangular) such that
\begin{equation}
G=BB^{\dagger} .
\end{equation} 
Then  the  corresponding vectors $g$ are given as  rows  of  $B$. $\square$
\end{proof}
\item[G(6)] Let $\nU(\lH)$ stands for the multiplicative group of unitary operators acting in~$\lH$. Then for any $U \in \nU(\lH)$, any $\Psi = \sum_{i=1}^d e_i \otimes \psi_i$:
\begin{equation}
G(\Psi) = G(( \mathbb{I} \otimes U)(\Psi)),
\end{equation}
i.e. $G$ is $U$-invariant.
\item[G(7)] Let $\sigma(\Psi) = (\lambda_1, \ldots, \lambda_d)$ be a spectrum of the matrix $G(\Psi)$ (ordered in  non-increasing  order) then $\lambda_i(\Psi) = (\lambda^{s}_{i}(\Psi))^2$ where $\lambda^{s}_{i}(\Psi)$ are the Schmidt coefficients of the vector $\Psi \in \mathbb{C}^d \otimes \lH$.
\item[G(8)] Let $\mathrm{Ent}(\Psi) = - \mtr{ G(\Psi) \log G(\Psi) }$.  Then $\sup_{\Psi} \mathrm{Ent}(\Psi)  = \mathrm{Ent}(\Psi^{\star}) = \log(d)$  where  $\Psi^{\star} \in \bC^d \otimes \lH$ is such that $\lambda_i = \frac{1}{d}$ for all $i=1:d$.
\end{itemize}

For the purposes of the present note $G(7)$ and $G(8)$  are the most important properties of the Gram matrix $G(\Psi)$. We conclude from these properties that in order to compute entropy of entanglement included in $\Psi$ we need to compute the corresponding Gram matrix of $\Psi$, matrix elements of it  are given as scalar products of the corresponding components $\psi_i$ of $\Psi$.

\begin{remark}
The scalar products are the simplest, widely used measures of similarity in between pure quantum states and are well known in quantum information theory as fidelity measure \cite{Nielsen2000} 
The Gram matrices are widely used in several areas of research. The differential geometry, mathematical statistics problems, quantum chemistry and atomic physics, control theory, machine learning and deep learning problems are some examples where the technique based on Gram matrieces are applied. See for a source of references on this \cite{Hazewinkel2001}.

However their use in Quantum Information Theory (QIT) seems to be underestimated as it is hard to find any reference in QIT in which an explicite use of Gram matrices (although they are very close to the corresponding reduced density matrices notion). For a more extensive review together with some (might be) new results obtained with the use of Gram matrices, see \cite{Gielerak2019a}, \cite{Christensen2003}.

However,  we  are  not  able  to  find  any  source  for  the  use  of  this  concept  in  the  present  context.
\end{remark}

Although we know the answer to the question which states $\Psi \in \mathbb{C}^d \otimes \lH$ are maximally entangled the problem how to read off the amount of entropy directly from the formula~(\ref{lbl:psi:decomposition:RG:MS:2018}) is the main topic of this contribution. The main result, confirmed numerically in the next section for the particular cases $d=2$ and $d=4$,  is the following hypothesis (for any pure state  $\Psi$  written in the form given by Eq.~(\ref{lbl:eq:Psi:form}) we define its Gram  matrix as the Gram matrix formed by the frame $(\psi_1, \psi_2, \ldots, \psi_d)$).

\begin{conjecture}
Let $\Psi = \sum_{i=1}^{d} e_i \otimes \psi_i \in \mathbb{C}^d \otimes \nH$. Then, $\Psi$ is maximally entangled pure state of $(d, \infty)$ quantum system under consideration iff:
\begin{equation}
{G(\Psi)}_{ij} = \frac{1}{d} \cdot \delta_{ij} \;\;\; \mathrm{for} \;\;\; i,j=1:d.
\end{equation}
\end{conjecture}

\begin{remark}
The statement "iff" from conjecture follows simply from the Schmidt decomposition Eq.~(\ref{lbl:eq:schmidt:RG:MS:2019}) and the results in \cite{Gielerak2018c}.
\end{remark}

\section{Numerical studies for $(2,\infty)$ and $(4, \infty)$ systems} \label{lbl:sec:title:RG:MS:CN:2019}

Let $G(\Psi)$ be a $d \times d$ Gram matrix for a given pure state $\Psi \in \bC^d \otimes \lH$ as given in formula (\ref{lbl:G:of:Psi}). The special Gram matrix $G_{max}$, defined as $(G_{max})_{ij}=\frac{1}{d} \cdot \delta_{ij}$ corresponds to case $\mbraket{\psi_i}{\psi_j} =\frac{1}{d}  \delta_{ij}$ for which the entropy of the Gram matrix attains maximal, possible values.

For this reason we introduce the following definition:

\begin{definition}
Deviation from the maximal entangled states manifold of a given state $\Psi \in \bC^d \otimes \lH$ is given as
\begin{equation}
\mathrm{dev}(\Psi) = \left( \sum_{1 \leq i \leq j \leq d} {|G(\Psi)_{ij} - (G_{max})_{ij}|}^2\right)^{\frac{1}{2}} .
\end{equation}
\end{definition}

Now we can state precisely our hypothesis:
\begin{equation}
\mathrm{If} \; \mathrm{dev}(\Psi) > 0 \; \mathrm{then} \; \mathrm{Ent}(\Psi) < \log(d).
\end{equation}

At the moment, no complete rigorous proof of this hypothesis is available. However, some computer assisted numerical analysis in the particularly interesting (from the point of view of quantum engineering) cases of dimensions $d=2$ (single qubit entangled with atom for example) and $d=4$ (relativistic $\frac{1}{2}$-spin degrees of freedom entangled with the orbital degrees of freedom in the relativistic, Dirac model of atom for example) are presented in the rest of this section.

\begin{example} The case of $d=2$. \\
Let
\begin{equation}
\Psi = e_1 \otimes f_1 + e_2 \otimes f_2 \in \mathbb{C}^2 \otimes \lH .
\end{equation}
The Gram matrix:
\begin{equation}
G(\Psi) = \left[ \begin{array}{cc}
{|| f_1 ||}^2 & \mbraket{f_1}{f_2} \\
\mbraket{f_2}{f_1} & {|| f_2 ||}^2
\end{array} \right],
\end{equation}
has the following structure (which is exactly of the form of reduced density matrix obtained  from  $\Psi$ by tracing out the degrees of freedom connected with the infinite dimensional subsystem):
\begin{equation}
G(\Psi) = \left[ \begin{array}{cc}
p & \overline{\sigma} \\
\sigma & (1 - p)
\end{array} \right],
\label{lbl:eq:Gram:Psi:for:Case:2}
\end{equation}
where $p={|| f_1 ||}^2 \in [0,1]$, $\sigma=\mbraket{f_1}{f_2}_{\lH}$ is a complex number such that $|\sigma| \leq \frac{1}{2}$ as it follows from the Cauchy-Schwartz inequality.

Solving the corresponding eigenvalue equation for $G(\Psi)$ the following eigenvalues of $G(\Psi)$ are obtained
\begin{equation}
\lambda_1 = \frac{1-\sqrt{\Delta}}{2}, \lambda_2 = \frac{1+\sqrt{\Delta}}{2}, \; where \; \Delta = 1-4(p-p^2-|\sigma|^2) .
\end{equation}

The corresponding entropic measure of quantum entanglement  contained in the state $\Psi$:
\begin{equation}
\mathrm{Ent}(\Psi) = - \lambda_1 \log(\lambda_1) - \lambda_2 \log(\lambda_2), 
\end{equation}
attains its maximal value for $\lambda_1 = \lambda_2 = \frac{1}{2}$ which corresponds to the case when $|\Delta| = 0$. In particular case $p=\frac{1}{2}$ and $|\sigma|=0$ the quantum state $\Psi$ is maximally entangled.
\end{example}

\begin{figure}
\begin{center}
\begin{tabular}{c}
(A) \\
\includegraphics[width=12cm]{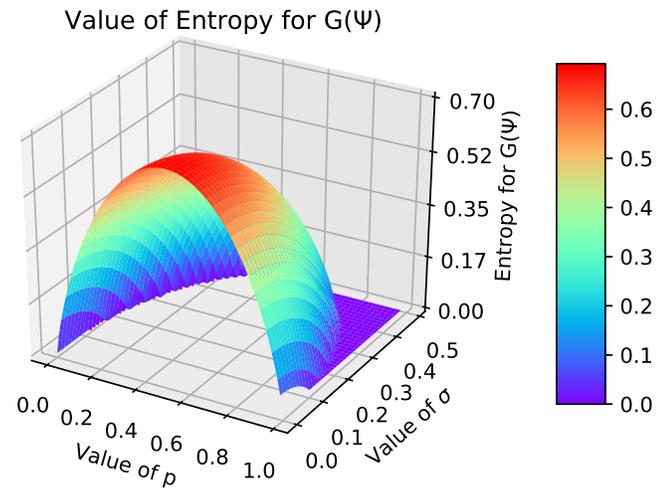} \\
(B) \\
\includegraphics[width=9cm]{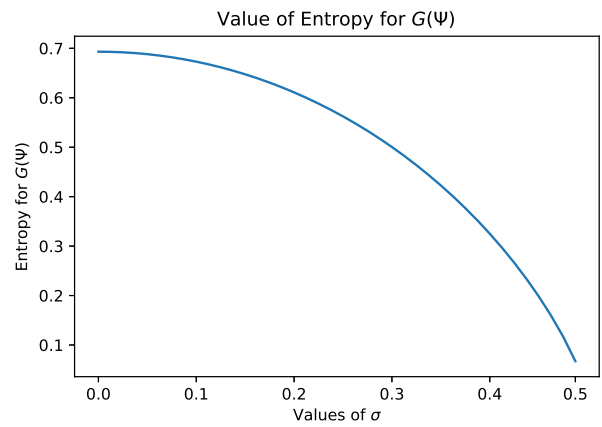} 
\end{tabular}
\end{center}
\caption{Value of Entropy measure $\mathrm{Ent}(\Psi)$ for case $d=2$. The parameters $p$ and $\sigma$ in plots (A) and (B) are defined as in Gram matrix $G(\Psi)$ denoted by Eq.~(\ref{lbl:eq:Gram:Psi:for:Case:2}). Plots (B) shows change of the Entropy measure value for $p=\frac{1}{2}$}
\end{figure}

\begin{example}
The case of $d=4$. \\
General form of the pure states in $d=4$ (corresponding to several physically interesting cases like photonic wave  functions for example) is given by
\begin{equation}
\Psi = e_1 \otimes f_1 + e_2 \otimes f_2 + e_3 \otimes f_3 + e_4 \otimes f_4,
\label{eq:Psi4:system}
\end{equation}
where $e_i \in \bC^4$ and $(e_i)_j = \delta_{ij}$, $f_i \in \lH$, $i=1:4$.
\end{example}

The corresponding to Eq.~(\ref{eq:Psi4:system}) Gram matrix $G(\Psi)$ is given by the formulae
\begin{equation}
G(\Psi)_{ij} = \mbraket{f_i}{f_j}_{\lH} 
\label{lbl:eq:G4:mat}.
\end{equation}
In  the  case  of  semi-definite  positive Gram matrix $G ( \Psi )$ the  well known  Sylvester  criterion, valid in the  case  of  strictly  positive matrices must  be  generalised. It is  necessary to  check in this  case that  all principal minors of $G$  are  non-negative \cite{Prussing1986}, \cite{Gilbert1991} not  only  the  leading  one. This  means  that in order to  describe  the  corresponding manifold  of  Gram matrices  containing  also the  matrices  with  non-trivial  kernels , $\dim( \mathrm{Ker}( G) ) > 0$  it  is necessary  to  put  $2^d-1$  non-linear  constrains of  the  form  as listed  in  Lemma~\ref{lbl:lemma:one:CN:2019:RG:MS}  below in the  case  of  strictly  positive Gram  matrices. In particular  dimension  $d=4$  this  number is  equal  to  15.This  is  the main  reason  that we  restrict  our  consideration  below  to  the  case  of  strictly positive  Gram  matrices only  which  are  most  interesting  from  the  physical point  of  view.

\begin{lemma}
Let $G_4$ be the set of $4 \times 4$ hermitian matrices $G \in \bC^2 \otimes \bC^2$ obeying the following constrains:
\begin{itemize}
\item [$G_4(1)$:] $\sum_{i=1}^{4} G_{ii} = 1$; $G_{ii} \in [0, 1]$,
\item [$G_4(2)$:] \begin{itemize}
\item[(i)] $\det(G) \geq 0$,
\item[(ii)] $\det\left( \left[\begin{array}{ccc} G_{11} & \ldots & G_{13} \\
\vdots & \ddots  & \vdots \\
G_{31} & \ldots & G_{33} 
 \end{array} \right] \right)  \geq 0$,
\item[(iii)] $\det\left( \left[\begin{array}{cc} G_{11} & G_{12} \\
G_{21} & G_{22} 
 \end{array} \right] \right)  \geq 0$.
\end{itemize}
\end{itemize}
Then to any $G \in G_{4}$ there exists a system of vectors $g_i \in \bC^4$ such that $\mbraket{g_i}{g_j}_{\bC^4} = G_{ij}$.
\label{lbl:lemma:one:CN:2019:RG:MS}
\end{lemma}
\begin{proof}
Essentially $G(1)$ together with $G(2)$ property give the normalisation and non-negativity of $G$. Thus applying property $G(5)$ of the present note we conclude the final statement of Lemma~\ref{lbl:lemma:one:CN:2019:RG:MS}. $\square$
\end{proof}


Our hypothesis is the following.

\begin{conjecture}
Let $G \in G_4$. Let us define the entropy functional on $G_4$:
\begin{equation}
\begin{array}{l}
\mathrm{Ent}: G_4 \longrightarrow [\sigma, \log(4) ], \\
G \longrightarrow \mathrm{Ent}(G) = - \mtr{G \log(G) } .
\end{array}
\end{equation}
Then $\sup_{G \in G_4} \mathrm{Ent}(G) = \log(4)$ and $\mathrm{Ent}(G) = \log(4)$ \myiff $G_{ij}=0$ for $i \neq j$ and $G_{ii}=1/4$ for $i=1:4$.
\end{conjecture}

\begin{figure}
\begin{center}
\includegraphics[width=12cm]{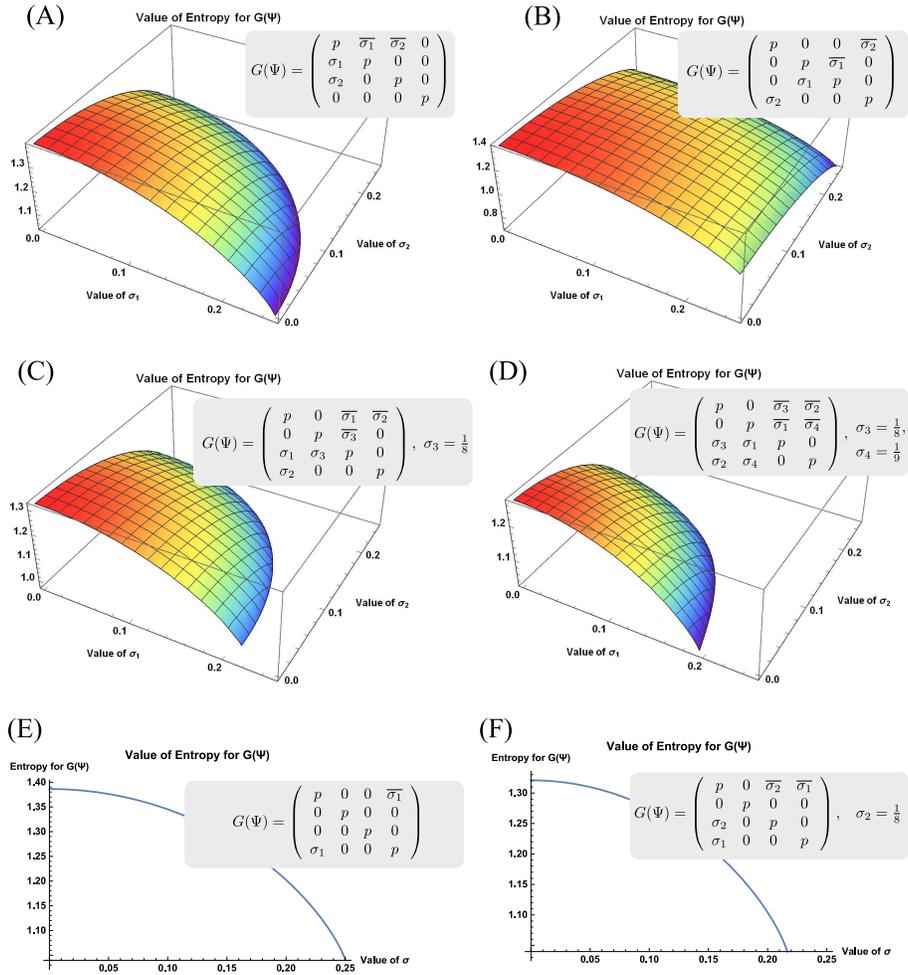}
\end{center}
\caption{Value of Entropy measure $\mathrm{Ent}(\Psi)$ for case $d=4$. The parameters $p=1/4$, and $\sigma_1$, $\sigma_2$, $\sigma_3$ represent elements of matrix $G(\Psi)$ denoted by Eq.~(\ref{lbl:eq:G4:mat}) in four form (A), (B), (C), (D) of $G(\Psi)$ matrix. Additionally, we also show two cases (E) and (F) where in (E) only one parameter $\sigma_1$ is used, and the second (F) where value $\sigma_2$ is constant}
\end{figure}


\section{Conclusions} \label{lbl:sec:conclusions:RG:MS:CN:2019}

A new method for entanglement analysis, based on the notion of unitary invariant Gram matrices connected with pure states of a general $(d,\infty)$ quantum systems has  been proposed. In particular, a new criterion for detecting the maximally possible amount of quantum entanglement present in an analysed
quantum state has been formulated in terms of the corresponding Gram matrix. Certain, numerical  analysis type  arguments in favour of the general validity of the  criterion formulated are being presented in a particular dimensions $d$, but very interesting cases of dimensions  $d=2$ and  $d=4$. The presented connection between  quantum entanglement and Gram matrices leads to a new, purely geometrical characterisation of entanglement in terms of a general and well known geometric interpretations of gramians.


\label{lbl:references:MS:JW:CN2019}


\begin{thebibliography}{99}



\bibitem{Nielsen2000} Nielsen, M.A., and Chuang, I.L.: Quantum Computation and Quantum Information, Cambridge University Press, Cambridge, U.K. (2000)

\bibitem{Bengtsson2006} Bengtsson, I.,  and \.Zyczkowski, K..: Geometry of Quantum States: An Introduction to Quantum Entanglement, Cambridge University Press, Cambridge, U. K. (2006)

\bibitem{Horodecki2009} Horodecki, R., Horodecki, P., Horodecki, M., Horodecki, K.: Quantum Entanglement, Rev. Mod. Phys., 81, 865,  available at arXiv: quantphys:/0702225 (2009)

\bibitem{Guhne2008} Guhne, O., Toth, G.: Entanglement detection, available at arXiv:quant-phys:/0811280 (2008)

\bibitem{Horodecki2010} Horodecki, R., Kilin, S.Y.,  Kowalik, J.S. (edts.): Quantum Cryptography and Computing. IOS Press (2010)

\bibitem{WitczakKrempa} Witczak-Krempa, W., Chen, G., Kim, Y.B., Balents, L.: Correlated quantum phenomena in the strong spin-orbit regime. Annual Review of Condensed Matter Physics 2014 Vol.~5, No.~1, pp.~57-82, also available at arXiv:1305.2193v2 (2013)

\bibitem{You2015}  You, W.L.,  Horsch, P.,  Ole\'s, A. M.: Quantum entanglement in the one-dimensional spin-orbital $SU(2) \otimes XXZ$ model. Phys. Rev. B 92, 054423 (2015)

\bibitem{Karimi} Karimi, E.,  Leach, J., Slussarenko, S., Piccirillo, B., Marrucci, L., Chen, L., She W., Franke-Arnold, S., Padgett, M. J., Santamato, E.: Spin-orbit hybrid entanglement of photons and quantum contextuality. Phys. Rev. A 82(2), 022115 (2011)

\bibitem{Gielerak2017} ---------------: Entangling qubit with the rest of the world - the monogamy principle in action. Studia Informatica, Vol. 38, No. 3, pp. 33--43 (2017)

\bibitem{Castelvecchi2017} Castelvecchi, D.: Quantum cloud goes commercial, Nature 543, 159 (2017)

\bibitem{IBMQ} IBM Makes Quantum Computing Available on IBM Cloud. https://www.research.ibm.com/ibm-q/ (2016)

\bibitem{Qiskit} Qiskit, https://github.com/IBM/qiskit-sdk-py (2019)

\bibitem{Rigetti} Rigetti, Rigetti Launches Quantum Cloud Services, Announces \$1Million Challenge. HPCwire. (2018)

\bibitem{Microsoft} https://www.microsoft.com/en-us/quantum/ (2019)

\bibitem{Google} https://ai.google/research/teams/applied-science/quantum-ai/ (2019)

\bibitem{Intel} Intel Invests US\$50 Million to Advance Quantum Computing. Intel Newsroom (2015)

\bibitem{DWave} D-Wave, https://www.dwavesys.com/our-company/meet-d-wave (2019)

\bibitem{Linke2017}  N. M. Linke, D. Maslov, M. Roetteler, S. Debnath, C. Figgatt, K. A. Landsman, K. Wright, and C. Monroe, Experimental comparison of two quantum computing architectures, PNAS 114, 3305-3310 (2017)

\bibitem{Devoret2004} Devoret, M. H., Wallraff, A., Martinis, J. M.: Superconducting Qubits: A Short Review.  arXiv:cond-mat/0411174 (2004)

\bibitem{DiVincenzo2000} DiVincenzo, D.P.:  The Physical Implementation of Quantum Computation. Fortschr. Phys., Vol. 48, pp. 771 -- 783 (2000)

\bibitem{Hazewinkel2001} Hazewinkel, M., (ed.) Gram matrix, in Encyclopaedia of Mathematics (2001)


\bibitem{Gielerak2018c} ---------------.: Schmidt decomposition of mixed-pure states for $(d, \infty)$ systems and some applications. arXiv:1803.09541 (2018)


\bibitem{Gielerak2019a} ---------------: in preparation (2019)


\bibitem{Christensen2003} Christensen, O.: An Introduction to Frames and Riesz Bases. Springer Science (2003)

\bibitem{Prussing1986} Prussing, J.E.: The principial minor test for semidefinite matrices. Journal of  Guidance, Control and Dynamics. Vol.~1, pp. 121 -- 122 (1986)

\bibitem{Gilbert1991} Gilbert, G.T.: Positive definite matrices and Sylvester criterion. The American Mathematical Monthly , Vol.~98, No.~1 , pp.~44 -- 46 (1991)

\bibitem{PhysRevLett.111.080502} Barends, R. and Kelly, J. and Megrant, A. and Sank, D. and Jeffrey, E. and Chen, Y. and Yin, Y. and Chiaro, B. and Mutus, J. and Neill, C. and O'Malley, P. and Roushan, P. and Wenner, J. and White, T. C. and Cleland, A. N. and Martinis, John M.: Coherent Josephson Qubit Suitable for Scalable Quantum Integrated Circuits. Phys. Rev. Lett. Vol. 111, Issue 8, pp. 080502 (2013)


\end{thebibliography}
\end{document}